\documentclass[journal,onecolumn]{IEEEtran}
\usepackage{amsmath,amsfonts,amssymb,amsthm,hyperref,enumerate,multicol,bm}
\usepackage{calligra,indentfirst,epsfig,lettrine}
\usepackage{relsize}
\usepackage{tabls}
\usepackage{array}
\usepackage{longtable}
\calligra
\DeclareFontFamily{U}{mathc}{}
\DeclareFontShape{U}{mathc}{m}{it}%
{<->s*[1.03] mathc10}{}

\DeclareMathAlphabet{\mathscr}{U}{mathc}{m}{it}

\usepackage{comment}
\makeatletter
\newcommand*{\rom}[1]{\expandafter\@slowromancap\romannumeral #1@}
\makeatother
\usepackage{mathtools}
\usepackage{graphicx}
\usepackage{subfigure}
\usepackage{xcolor}
\usepackage{color, soul}
\usepackage{mathrsfs}
\DeclareMathAlphabet{\mathpzc}{OT1}{pzc}{m}{it}

\usepackage{upgreek}

\newtheorem{definition}{Definition}
\newtheorem{corollary}{Corollary}
\newtheorem{theorem}{Theorem}
\newtheorem{example}{Example}
\newtheorem{lemma}{Lemma}
\newtheorem{remark}{Remark}

\allowdisplaybreaks
\usepackage[table]{xcolor}

\begin{document}
\title{Recovery Set Structures and Service Rates of Codes Obtained by the Plotkin-type Construction}

\author {Priyanka~Choudhary
       and~Maheshanand~Bhaintwal\\ 
       
       \thanks{The authors are with the Department of Mathematics, Indian Institute of Technology Roorkee, Roorkee 247667, India (e-mails: \href{mailto:priyanka_c@ma.iitr.ac.in}{priyanka\_c@ma.iitr.ac.in},
   \href{mailto:maheshanand@ma.iitr.ac.in}{maheshanand@ma.iitr.ac.in}).}}

\date{}
\maketitle
\begin{abstract}


 In distributed storage systems, redundancy enables a single data object to be reconstructed using several disjoint groups of servers. The resulting service rate region (SRR) captures all combinations of object request rates that the system can support simultaneously without overloading any individual server. We analyze the SRR of binary codes obtained via the Plotkin-type construction. We demonstrate that each recovery set for an information object of the Plotkin-type code \(\mathcal{C}_p\) induces a corresponding recovery set for the same data object characterized by the underlying code \(\mathcal{C}\). Furthermore, by characterizing the recovery sets structure of \(\mathcal{C}_p\) in terms of the recovery set structure of $\mathcal{C}$, we identify the additional recovery sets introduced by this construction. Using the associated recovery hypergraphs, we establish bounds on the SRRs of \(\mathcal{C}_p\) and iterated code \(\mathcal{C}_{p^m}\) in terms of the SRR of $\mathcal{C}$. 
 We consider the family of first-order binary Reed–Muller codes and show how the recovery structure, and consequently, the service rates of R\((1,m)\) for arbitrary $m$ can be recursively derived from a generator matrix of R\((1,2)\) through our results for successive Plotkin-type constructions.

\end{abstract}


\IEEEpeerreviewmaketitle

\section{Introduction}

\IEEEPARstart{T}he increasing demand for concurrent data access in distributed storage systems (DSS) has motivated the study of coding schemes that can provide both reliability and efficient service. In this context, the \emph{service rate region} (SRR) captures the collection of request rate vectors that a storage system can support simultaneously and therefore provides a natural way to quantify its service capability \cite{alfarano2024service}.
The SRR of a code characterizes the set of request rate vectors that can be simultaneously supported when the coordinates of a codeword are stored on individual servers. In this setting, a request for a data object may be served from any of its recovery sets, and the achievable service rates depend not only on the code parameters but also on the combinatorial structure of those sets. Consequently, understanding the recovery-set structure of a code is fundamental to analyzing its SRR. 


 For linear codes, the SRR is closely tied to the combinatorial structure of their recovery sets and can be viewed as a convex polytope in \(\mathbb{R}^k\) \cite{alfarano2024service}. This connection has led to detailed studies of SRRs for several classes of codes, including MDS codes \cite{Aktas, oval, LyH}, binary simplex codes \cite{Combit, Geo}, batch codes \cite{LRC}, Hamming codes \cite{HamPC}, and first-order Reed--Muller codes \cite{Geo}. The role of the structure of the dual of a code and its generator matrix in bounding achievable service rates was investigated in \cite{Alfarano}. Related works have also explored bounds for the maximum achievable service rate for an individual data symbol using general linear coding schemes, and whether these bounds can be achieved when there are combinatorial designs in the supports of dual codewords \cite{Ly, PCM}. Despite these developments, determining the SRR becomes increasingly difficult as the recovery set structure grows in complexity, particularly when the recovery graphs/hypergraphs cannot be explicitly enumerated. Recently, Ly, Soljanin, and Lalitha \cite{LyHS} studied the recovery structure of Reed--Muller codes and obtained bounds on the maximum service rate of information symbols. Then, Yu et al. \cite{Yu} refined the existing SRR analysis by providing an explicit characterization of the intersection patterns of recovery sets for high-order Reed--Muller codes. These results highlight the strong interplay between the algebraic structure of a code and its achievable service rates, and motivate further investigation into how specific code constructions transform recovery sets and, consequently, their SRRs.
A parallel line of work addresses the code designing problem, i.e., constructing codes whose service regions attain given structural requirements or contain a specific demand vector \cite{kurz,kilicc}.

The hypergraph-based approach \cite{LyH, HamPC} indicates that the collection of recovery sets associated with an information object naturally defines a recovery hypergraph. This representation provides a useful framework for studying the service capabilities of a code. In particular, matchings yield collections of pairwise disjoint recovery sets that can serve requests simultaneously, whereas fractional matchings allow recovery capacity to be distributed among intersecting recovery sets. 
The corresponding fractional vertex cover problem, through linear programming duality, provides an equivalent characterization and useful upper bounds on the achievable service rates. 

What is largely absent from the literature is a complementary perspective: rather than computing the region of one code family at a time, one may ask how the region transforms when a standard code construction operator is applied to a code whose recovery structure or, consequently, the SRR is already known. The Plotkin, or $(\mathbf u\mid\mathbf u+\mathbf v)$, construction is the natural initial test case: it gives the elementary step in the recursive description of Reed--Muller codes and doubles the number of servers while adding only one new information object. Thus, any gain it offers reflects increased flexibility rather than improved storage efficiency.

In this work, we investigate the effect of the Plotkin construction on the recovery structure and SRR of a code. The Plotkin construction considered here is a recursive construction that constructs a code of length $2n$ from two constituent codes, namely any arbitrary code $\mathcal C$ of length $n$ and a repetition code of length $n$.
We determine the minimal recovery sets for the information objects of the Plotkin-type code
$\mathcal{C}_p=\{(\mathbf c,\mathbf c):\mathbf c\in\mathcal C\}\cup
\{(\mathbf c,\mathbf 1_n+\mathbf c):\mathbf c\in\mathcal C\}$ in terms of the
recovery structure of that object in $\mathcal {C}$. Our main structural result shows that the map
$R_p\mapsto R_p^{(1)}\triangle R_p^{(2)}$ sends every minimal recovery set corresponding to the generator matrix
$\mathcal{G}_p$ to a recovery set corresponding to the generator matrix $\mathcal{G}$. 
Our work not only enumerates all recovery sets associated with Plotkin-type code but also determines bounds on the service rates of codes obtained through recursive applications of the Plotkin construction, and the resulting service rate region becomes explicitly computable. 
 As an application, we examine the family of first-order binary Reed--Muller codes. The codes $R(1,m)$ can be described recursively using the Plotkin construction, building upon lower-order codes of the family.  Therefore, beginning with the recovery structure associated with a generator matrix of $R(1,2)$,  our results provide a recursive method to identify the recovery sets and analyze the service rates of $R(1,m)$ for arbitrary $m\geq 2$. Our results are consistent with the recovery-set description obtained in \cite{LyHS}. Additionally, the axis intercepts agree with the values predicted by \cite{Ly,LyHS}.
 


The remainder of the paper is organized as follows. Section~\ref{sec2} introduces the necessary preliminaries required to read this paper. In Section~\ref{sec3}, we characterize the recovery sets arising from the Plotkin-type construction and establish the relationship between the recovery structures of $\mathcal{C}$ and $\mathcal{C}_p$. We give a complete characterization of the minimal recovery sets of an information object in 
$\mathcal{C}_p$ in terms of its recovery sets in the parent code $\mathcal{C}$ and the supports of codewords of $\mathcal{C}^\perp$. Section~\ref{sec4} studies the corresponding recovery hypergraphs and derives bounds on the SRR through matching and fractional covering techniques. Starting from an optimal fractional vertex cover and an optimal fractional matching of the recovery hypergraph associated with $\mathcal{C}$, we derive bounds on the fractional vertex cover number and matching numbers of the corresponding hypergraph for $\mathcal{C}_p$, and thereby bounds on the service rates of the associated data symbols. Furthermore, we recursively define the code \(\mathcal{C}_{p^m}\) and derive bounds on the SRR of this code as well. 
 Within this section, we apply these results to first-order binary Reed--Muller codes, providing a recursive characterization of their service rates starting from R\((1,2)\). 
 Finally, Section~\ref{secC} concludes the paper.


\section{Preliminaries}\label{sec2}

In this section, we present the notation and terminology used throughout the paper. We first recall the basic notions of linear codes, and then introduce recovery sets and their associated hypergraph representation. 
Next, we recall the definition of the service rate region of a coded distributed storage system. Finally, we
describe the Plotkin-type construction examined in this work, along with the corresponding generator matrix. 
\subsection{Linear Codes}

For positive integers \(a\) and \(b\) with \(a\leq b\), let $[a,b]=\{a,a+1,\ldots,b\}$,
and, in particular, let
$[1,n]=[n]=\{1,2,\ldots,n\}.$ 
Let \(q\) be a prime power and let \(\mathbb{F}_q\) denote the finite field
with \(q\) elements. The all-zero and all-one vectors of length \(n\) are
denoted by \(\mathbf{0}_n\) and \(\mathbf{1}_n\), respectively. The
\(i\)-th standard basis vector of \(\mathbb{F}_q^k\) is denoted by
\(\mathbf e_i\). For a set $S$, $|S|$ denotes the cardinality of $S$.

A \(q\)-ary linear \([n,k,d]\) code \(\mathcal C\) is a
\(k\)-dimensional subspace of \(\mathbb{F}_q^n\) with minimum Hamming distance \(d\). A $k\times n$ matrix $\mathcal G\in\mathbb{F}_q^{k\times n}$ of rank $k$ whose rows form a basis of \(\mathcal C\) is called a
generator matrix of \(\mathcal C\). For an information (message) vector $\mathbf x=(x_1,\ldots,x_k)\in\mathbb{F}_q^k,$  the corresponding codeword in \(\mathcal C\) is
$\mathbf c=\mathbf x\cdot\mathcal \mathcal G$. 
We denote the \(j\)-th column of \(\mathcal G\) by
\(\mathbf g_j\), \(j\in[n]\). A column \(\mathbf g_j\) is called systematic for  $i\in[k]$ if \(\mathbf g_j=\alpha e_i\), for some nonzero $\alpha\in\mathbb{F}_q$. If for every $i\in[k]$, there exists a systematic column in $\mathcal{G}$, then $\mathcal{G}$ is known as a systematic generator matrix for $\mathcal{C}.$ Without loss of generality, throughout this paper, we consider the standard form for a systematic generator matrix,
$ \mathcal G=(I_k \mid A),$
where $I_k$ is the $k\times k$ identity matrix.

For a vector $\mathbf c=(c_1,\ldots,c_n)\in\mathbb{F}_q^n$, its support is defined by
$\operatorname{supp}(\mathbf c)
=
\{j\in[n]:c_j\neq 0\}$.
The dual code of \(\mathcal C\) is defined by
\[
\mathcal C^\perp
=
\left\{
\mathbf h\in\mathbb F_q^n:
\mathbf h\cdot \mathbf c^\top=0
\text{ for every }\mathbf c\in\mathcal C
\right\},
\]
where $(\cdot)$ denotes the usual Euclidean product. The dual code $\mathcal{C}^\perp$ is an \([n,n-k,d^\perp]\) linear code, where
$d^\perp=d(\mathcal C^\perp)$ denotes the minimum distance of $\mathcal C^\perp$. Let $\mathcal{H}$ be a generator matrix of $\mathcal{C}^\perp$, then $\mathcal{G}\cdot\mathcal{H}^\top=0.$ The matrix $\mathcal{H}$ is called a parity-check matrix of $\mathcal{C}.$

A vector \(\mathbf h=(h_1,\ldots,h_n)\in\mathcal C^\perp\) induces a linear
dependence relation among the entries $c_j$ for every $\mathbf c=(c_1,\ldots,c_n)\in\mathcal C$, i.e.,
\[
\sum_{j=1}^n h_j c_j=0.
\]

\subsection{Recovery Sets and Hypergraph Representation}

Consider a coded data storage system consisting of \(n\) servers storing $k$ information symbols,
where server \(j\) stores the \(j\)-th coordinate of the encoded vector
\(\mathbf x\cdot \mathcal G\). The \(i\)-th information symbol $x_i$ is represented by the
standard basis vector \(\mathbf e_i\).

\begin{definition}
Let \(i\in[k]\). A set \(R\subseteq[n]\) is called a \emph{recovery set}
for the \(i\)-th information symbol if
$\mathbf e_i
\in
\operatorname{span}\{\mathbf g_j:j\in R\}.$

\end{definition}

 A recovery set \(R\) for an information symbol $x_i$
is called \emph{minimal} if no proper subset of \(R\) is a
recovery set for \(x_i\).
We denote the collection of all minimal recovery sets  for
$x_i$  corresponding to the generator matrix $\mathcal{G}$ by $\mathcal R_i^{\min}(\mathcal G)=\mathcal R_i(\mathcal G)$.




For a systematic generator matrix $\mathcal{G}$, a set
\(R\subseteq[n]\setminus\{i\}\) is a non-singleton recovery set for \(x_i\) if
and only if there exists a codeword
\(\mathbf v\in\mathcal C^\perp\) such that
$i\in\operatorname{supp}(\mathbf v)
\subseteq R\cup\{i\}$ \cite{Alfarano}. 

This characterization links the recovery
sets directly to the supports of codewords of the dual code. Let $\mathbf{h}\in\mathcal{C}^\perp$, we say \(D=\operatorname{supp}(\mathbf{h})\) is an
inclusion-minimal support if \(D\) contains
no proper subset whose corresponding columns of $\mathcal{G}$ sum to zero.
 Equivalently, \(D\) is a minimal linear dependency among the columns of \(\mathcal G\). The term dual support means that the given set is the support of a codeword of $\mathcal{C}^\perp.$

The recovery set structure of a linear code can be represented by a
hypergraph.

\begin{definition}
For a generator matrix $\mathcal{G}$ of a linear code $\mathcal{C}$, define the recovery hypergraph $\Gamma_{\mathcal G}
=
\bigl([n],\mathcal{E})$, 
whose vertex set is the set of storage nodes \([n]\), and $\mathcal{E}$ is the set of all hyperedges. Here, $\mathcal{E}$
is the collection of all the minimal recovery sets in $\mathcal R_i(\mathcal G)$ for all $i\in[k].$
\end{definition}
All hyperedges corresponding to $\mathcal R_i(\mathcal G)$ are labeled $\mathbf{e}_i.$
We may also define the partial recovery hypergraph of any $I\subseteq[k]$ as
$\Gamma_{\mathcal G}(I)
=
\bigl([n],\mathcal E(I)\bigr)$,
where $\mathcal E(I)$ consists of  hyperedges only in $\mathcal{R}_i(\mathcal{G})$ for all $i\in I,$ i.e., all hyperedges labeled $\mathbf{e}_j$ are removed from $\Gamma_{\mathcal G}$, where $j\in[k]\setminus I
.$
 Throughout this paper, whenever $I=\{a\}\subseteq [k],$ we adopt the notation $\Gamma_{\mathcal{G}}^{a}$ for the partial recovery hypergraph.

 A \emph{matching} in hypergraph 
$\Gamma_{\mathcal{G}}$ is a set ($\subseteq\mathcal{E}$) of pairwise
vertex-disjoint hyperedges. The \emph{matching number} of $\Gamma_{\mathcal{G}}$,
denoted by $\nu(\Gamma_{\mathcal{G}})$, is the maximum cardinality of a matching. Equivalently,
\[
\nu(\Gamma_{\mathcal{G}})
=
\max
\left\{
\sum_{E\in\mathcal{E}}\alpha_E:
\sum_{E\ni v}\alpha_E\leq 1\ \forall v\in[n]
\right\}, 
\]
where $\alpha_E\in\{0,1\}$ for all $E\in\mathcal{E}$. A \emph{vertex cover} of $\Gamma_{\mathcal{G}}$ is a subset $S\subseteq[n]$ that intersects every hyperedge. The \emph{vertex cover number}, denoted by $\tau(\Gamma_{\mathcal{G}})$, is the minimum
cardinality of a vertex cover, and can equivalently be written as
\[
\tau(\Gamma_{\mathcal{G}})
=
\min
\left\{
\sum_{v\in [n]}\beta_v:
\sum_{v\in E}\beta_v\geq 1\ \forall ~E\in\mathcal{E}
\right\},
\]
where $\beta_v\in\{0,1\}\ \forall ~v\in [n]$. Relaxing the integrality constraints gives the corresponding fractional notions.
A \emph{fractional matching} is an assignment
$\gamma:\mathcal{E}\to[0,1]$ satisfying
\[
\sum_{E\ni v}\gamma_E\leq1
\qquad\text{for every }v\in [n].
\]
The \emph{fractional matching number}, denoted by $\nu_f(\Gamma_{\mathcal{G}})$, is the
maximum total weight of such an assignment, i.e., $\nu_f(\Gamma_{\mathcal{G}})
=
\max
\left\{
\sum_{E\in\mathcal{E}}\gamma_E
\right\}$.
Similarly, a \emph{fractional vertex cover} is an assignment
$\delta:[n]\to[0,1]$ satisfying
\[
\sum_{v\in E}\delta_v\geq1
\qquad\text{for every }E\in\mathcal{E},
\]
and the \emph{fractional vertex cover number}, denoted by $\tau_f(\Gamma_{\mathcal{G}})$,
is given by $\tau_f(\Gamma_{\mathcal{G}})
=
\min
\left\{
\sum_{v\in [n]}\delta_v
\right\}$.
\subsection{Service Rate Region and Fractional Matchings}

Consider a coded DSS with \(n\) servers storing
\(k\) information symbols using a generator matrix \(\mathcal G\). We assume
that each server has unit service capacity. Thus, server \(j\) can process
requests at an average rate of at most $\mu_j=1,
 j\in[n]$.

Let
$\boldsymbol{\lambda}
=
(\lambda_1,\ldots,\lambda_k)
\in\mathbb R_{\geq0}^k$ 
denote a request-rate vector, where \(\lambda_i\) is the rate at which
requests for the \(i\)-th information symbol arrive. For each \(i\in[k]\), let
$\mathcal R_i(\mathcal G)
=
\{R_{i,1},\ldots,R_{i,t_i}\}$ 
denote the collection of minimal recovery sets for $i$-th information symbol. A scheduling policy distributes the requests for each information symbol
among its recovery sets. Let $\lambda_{i,\ell}\geq0$
denote the fraction of $\lambda_i$ assigned to 
\(R_{i,\ell}\).

\begin{definition}
The \emph{service rate region} of \(\mathcal G\), denoted by
\(\Lambda(\mathcal G)\), is the set of all
\(\boldsymbol{\lambda}\in\mathbb R_{\geq0}^k\) for which there exists a
nonnegative allocation
$\{\lambda_{i,\ell}:
i\in[k],\ \ell\in[t_i]\}$
satisfying
\begin{equation}
\sum_{\ell=1}^{t_i}\lambda_{i,\ell}
=
\lambda_i,
\qquad
i\in[k],
\end{equation}
and
\begin{equation}
\sum_{i=1}^k
\sum_{\ell:\,j\in R_{i,\ell}}
\lambda_{i,\ell}
\leq 1,
\qquad
j\in[n].
\end{equation}
\end{definition}

Constraint (1) ensures that the complete request rate for each information
symbol is served. Constraint (2) ensures that the aggregate request rate
assigned to any server does not exceed its unit service capacity.
An allocation satisfying (1) and (2) is called a \emph{feasible allocation}
for \(\boldsymbol{\lambda}\). 
A request rate vector belonging to
\(\Lambda(\mathcal G)\) is called \emph{achievable}. 
Let $I\subseteq[k]$, then define \[\left(\sum_{i\in I }\lambda_i\right)^*=\max\left\{\sum_{i\in I }\lambda_i: \boldsymbol{\lambda}=(\lambda_1,\ldots,\lambda_k)\in\Lambda(\mathcal{G})\right\}.\]
 For $I=\{i\}$, $\lambda_i^*$ is called the maximal achievable service rate for the $i$-th information symbol.
Since every recovery set contains at least one minimal recovery set, it is
sufficient to consider inclusion-minimal recovery sets when determining the SRR. 


The SRR of a coded DSS admits a natural combinatorial interpretation through the recovery hypergraph associated with a generator matrix of the code. Let \(\mathcal{C}\) be an \([n,k,d]\) linear code with a generator matrix $\mathcal{G}$, and let \(\Gamma_\mathcal{G}=([n],\mathcal E)\) denote its recovery hypergraph.

Consider a request vector $\boldsymbol{\lambda}=(\lambda_1,\lambda_2,\ldots,\lambda_k)$.  If \(w_E\) denotes the amount of demand assigned to the recovery set represented by the hyperedge \(E\in \mathcal E\), then the node-capacity constraints are given by
\[
\sum_{E\ni v}w_E\leq 1,\qquad \forall\,v\in[n]
.\]

As shown in \cite{LyH, Combit, Aktas}, these are precisely the feasibility constraints of a fractional matching on $\Gamma_G$. Consequently, every fractional matching $\mathbf{w}\in \mathbb{R}^{|\mathcal E|}_{\geq 0}$ induces a feasible service allocation $\boldsymbol{\lambda}(\mathbf{w})=(\lambda_1(\mathbf{w}), \lambda_2(\mathbf{w}), \ldots, \lambda_k(\mathbf{w}))$, where $\lambda_i(\mathbf{w})$ is the total sum of $w_E$ over all edges labeled by $\mathbf e_i$. Conversely, every feasible service allocation corresponds to a fractional matching \cite[Proposition 1]{Aktas}. 
In other words, the SRR is the image of the fractional matching polytope of the recovery hypergraph under a linear mapping that associates each fractional matching with its corresponding request rate vector.

Under the fractional matching problem formulation of SRR \cite{Combit}, the maximum achievable service rate is characterized by the fractional matching number of the recovery hypergraph. This, together with the classical matching--vertex cover inequality gives
\begin{equation}\label{inequalities}
    \nu(\Gamma_{\mathcal{G}})\leq \left(\sum\limits_{i=1}^k\lambda_i\right)^* =\nu_f(\Gamma_\mathcal{G})=\tau_f(\Gamma_\mathcal{G})\leq\tau(\Gamma_\mathcal{G}).
\end{equation}

 It can easily be verified that the relation given in \eqref{inequalities} is also true for any partial recovery hypergraph $\Gamma_{\mathcal{G}}(I)$, where $I\subset[k].$

\subsection{Codes and Their Plotkin-Type Construction}\label{plotkinc}
 Let $\mathcal C$ be an \([n,k,d]\) binary linear code with a generator matrix and a parity-check matrix
$\mathcal G$ and $\mathcal{H}$, respectively.

We consider the Plotkin-type construction \cite{Huffman}
\[
\mathcal C_p
=
\left\{
(\mathbf c,\mathbf c):\mathbf c\in\mathcal C
\right\}
\cup
\left\{
(\mathbf c,\mathbf 1_n+\mathbf c):\mathbf c\in\mathcal C
\right\}
\label{C_p}.
\]
$\mathcal C_p$ is a $[2n, k+1, \min\{n,2d\}]$ linear code over \(\mathbb F_2\). 

A generator matrix and a parity-check matrix for \(\mathcal C_p\) is given by $$\mathcal{G}_p=\begin{pmatrix}
    \mathcal{G} &\mathcal{G}\\\mathbf{0}_n &\mathbf{1}_n
\end{pmatrix},\qquad \text{ and }\qquad \mathcal{H}_p=\begin{pmatrix}
    \mathcal{H} &\mathbf{O}\\\mathcal{H}^\prime&\mathcal{H}^\prime
\end{pmatrix}, $$ respectively, where $\mathbf{O}$ stands for the $(n-k)\times n$ zero matrix, and $\mathcal{H}^\prime$ is a $(n-1)\times n$ parity-check matrix of the repetition code. A standard systematic choice for $\mathcal{H}^\prime$ is $\begin{pmatrix}
    I_{n-1} & \mathbf{1}_{n-1}^\top
\end{pmatrix}.$ Hence, $ d(\mathcal{C}_p^\perp)=\min\{d(\mathcal{C}^\perp),4\}.$

Throughout the paper, unless otherwise stated, $\mathcal C$ denotes a binary linear
code, and $\mathcal{C}_p$ denotes the associated Plotkin-type code defined in \ref{C_p}. Their respective generator matrices are denoted by $\mathcal{G}$ and $\mathcal{G}_p$ as defined in this section. Further, $\mathcal G$ has no zero column and no repeated systematic columns for any data object.
 
\section{Recovery set structure augmentation for Plotkin-type code}\label{sec3}

Let $\mathbf{x}=(x_1, x_2, \ldots,x_k)$ denote an information vector containing $k$ information symbols for a codeword $\mathbf{c}=(c_1, c_2, \ldots, c_n)\in \mathcal{C}$, and let $x_{k+1}$ be the newly added $(k+1)$-th information symbol for the associated codeword $\mathbf{c}_p=(c_1, c_2, \ldots, c_n, c_{n+1}, c_{n+2}, \ldots, c_{2n})\in\mathcal{C}_p$, where $c_{n+j}=c_j+x_{k+1}$ for all $j\in[n]$.

For a subset \(R\subseteq[n]\), we define
$R+n
:=
\{j+n:j\in R\}.$ For subsets \(A,B\subseteq[2n]\), their symmetric difference is defined by $A\triangle B
:=
(A\setminus B)\cup(B\setminus A)$.

It follows immediately from the definition that for any  $i\in[k]$ of $\mathcal{C}$, we have $\mathcal{R}_i(\mathcal{G})\subseteq\mathcal{R}_i(\mathcal{G}_p).$ Therefore, our main focus is on the recovery sets present in $\mathcal{R}_i(\mathcal{G}_p)\setminus \mathcal{R}_i(\mathcal{G})$. Now, we characterize the relationship between a recovery set
$R_p\in\mathcal{R}_i(\mathcal{G}_p)$ and the corresponding recovery set
induced in $\mathcal{R}_i(\mathcal{G})$.

Let $i\in[k]$ and let $R_p\in\mathcal{R}_i(\mathcal{G}_p)$, i.e., for $\mathbf{c}_p=(c_1,\ldots,c_n,$ $c_{n+1},\ldots, c_{2n})\in\mathcal{C}_p$, we have \[x_i=\sum_{j\in R_p}c_j.\]
Define
\[
R_p^{(1)}:=R_p\cap[n],
\quad
R_p^{(2)}:=\{j-n:j\in R_p\cap[n+1,2n]\}.
\]Note that, if $R_p^{(2)}=\emptyset$, then $R_p^{(1)}=R_p\in\mathcal{R}_i(\mathcal{G}).$ Therefore, we examine only those recovery sets in $\mathcal{R}_i(\mathcal{G}_p)$ with $R_p^{(2)}\neq \emptyset.$
Now \begin{align*}
    x_i&=\sum_{j\in R_p^{(1)}}c_j+\sum_{j^\prime\in R_p^{(2)}}c_{n+j^\prime}=\sum_{j\in R_p^{(1)}}c_j+\sum_{j^\prime\in R_p^{(2)}}(c_{j^\prime}+x_{k+1})\\&=\sum_{j\in R_p^{(1)}\triangle R_p^{(2)}}c_j+|R_p^{(2)}|x_{k+1},\end{align*}
 where $\triangle$ denotes the symmetric difference.
Since we want to induce recovery sets associated with $\mathcal{G}$, there should be no role of $x_{k+1}$. Therefore, $|R_p^{(2)}|$ must be even.

\begin{remark}\label{cap12}
    The minimality of $R_p$ forces $|R_p^{(1)}\cap R_p^{(2)}|\leq 1.$ Because if there are two distinct elements $m$ and $m^\prime$ in $R_p^{(1)}\cap R_p^{(2)}$, then $\bigl(R_p^{(1)}\setminus\{m,m^\prime\}\bigr)\cup \bigl((R_p^{(2)}\setminus\{m,m^\prime\})+n\bigr)$ also recovers $x_i$, contradicting minimality of $R_p.$
\end{remark}

Moreover, if there exists $\mathbf{h}\in\mathcal{C}^\perp$ such that $\operatorname{supp}(\mathbf{h})=D\subsetneq R_p^{(1)}\triangle R_p^{(2)}$, then we have
\[x_i=\sum_{j\in (R_p^{(1)}\triangle R_p^{(2)})\setminus D}c_j.\]
Hence, $(R_p^{(1)}\triangle R_p^{(2)})\setminus D$ is a recovery set for $x_i$. 

\begin{remark}\label{2D}
    It is important to note that if such a $D$ exists, then $|R_p^{(2)}\cap D|$ is odd. Because, if $|R_p^{(2)}\cap D|$ is even then $\big(R_p^{(1)}\setminus D\big)\cup\left(\big(R_p^{(2)}\setminus D\big)+n\right)$ recovers $x_i$, and it is contained in $R_p$. This contradicts the minimality of $R_p.$ \end{remark}

Now, we demonstrate that $(R_p^{(1)}\triangle R_p^{(2)})\setminus D$ is a minimal recovery set for $x_i$. Suppose, for contradiction, that there exists a nonempty proper subset $T \subsetneq \left(R_p^{(1)} \triangle R_p^{(2)}\right)\setminus D$
such that $T\in\mathcal{R}_i(\mathcal{G})$. Let
$T=A\cup B,$ 
where $A\subseteq R_p^{(1)}\setminus (R_p^{(2)}\cup D)$ and $B\subseteq R_p^{(2)}\setminus (R_p^{(1)}\cup D)$. 

First, suppose that there does not exist a nonempty subset $D\subsetneq R_p^{(1)}\triangle R_p^{(2)}$ such that $D=\operatorname{supp}(\mathbf{h})$ for some $\mathbf{h}\in\mathcal{C}^\perp$. Then
    \[\big( R_p^{(1)}\setminus A\big) \cup A\cup (B+n)\cup \big((R_p^{(2)}\setminus B)+n\big)=R_p.\] We have \begin{align}x_i&=\sum_{j\in R_p}c_j=\sum_{j\in R_p^{(1)}\setminus A}c_j+\sum_{j\in A}c_j+\sum_{j\in B}(c_j+x_{k+1})+\sum_{j\in R_p^{(2)}\setminus B}(c_j+x_{k+1}), \\&=\sum_{j\in R_p^{(1)}\setminus A}c_j+\sum_{j\in R_p^{(2)}\setminus B}c_j+\sum_{j\in T}c_j+|R_p^{(2)}|x_{k+1},\\&=\sum_{j\in R_p^{(1)}\setminus A}c_j+\sum_{j\in R_p^{(2)}\setminus B}c_j+x_i.\end{align}

Moreover, since $T $ is a proper subset of $R_p^{(1)} \triangle R_p^{(2)}$, $\big(R_p^{(1)}\setminus A\big)\triangle\big( R_p^{(2)}\setminus B\big)$ is nonempty. This implies that $\big(R_p^{(1)}\setminus A\big)\triangle\big( R_p^{(2)}\setminus B\big)$ is the support of some codeword of $\mathcal{C}^\perp$. This contradicts precisely the assumption that no such \(D\) exists. Therefore, in this case no proper subset $T$ exists that recovers $x_i$, implying that $\big(R_p^{(1)}\setminus A\big)\triangle\big( R_p^{(2)}\setminus B\big)\setminus D$ is minimal recovery set.

Next, assume that there exists a nonempty subset $D=\operatorname{supp}(\mathbf{h})\subsetneq R_p^{(1)}\triangle R_p^{(2)}$ for some $\mathbf{h}\in\mathcal{C}^\perp$. We consider two cases according to the parity of $|B|$.

\begin{enumerate}[(a).]
    \item $|B|$ is even. Then 
it can easily be verified that $A\cup (B+n)$ is also a recovery set for $x_i$. 
It follows that
$A\cup (B+n)\subsetneq R_p.$
Thus, $R_p$ contains a proper subset that is itself a recovery set for $x_i$, contradicting the minimality of $R_p$.

 \item $|B|$ is odd. Then  $A\cup(B+n)\cup \big(R_p^{(1)}\cap D \big)\cup \left(\big(R_p^{(2)}\cap D\big)+n\right)\subsetneq R_p$, and recovers $x_i$. Again, contradicting the minimality of $R_p$. 
\end{enumerate} 


Hence, no such set $T$ exists. Consequently, $(R_p^{(1)}\triangle R_p^{(2)})\setminus D\in \mathcal{R}_i(\mathcal{G}).$ In other words,  $R_p^{(1)}\triangle R_p^{(2)}\supseteq R$ for some $R\in\mathcal{R}_i(\mathcal{G})$.
\begin{remark}\label{rem:D-unique}
Let $i\in[k]$ and let $R_p\in\mathcal{R}_i(\mathcal{G}_p)$. If there exists a nonempty subset
$D\subseteq R_p^{(1)}\triangle R_p^{(2)}$ such that $D=\operatorname{supp}(\mathbf{h})$ for some
$\mathbf{h}\in\mathcal{C}^{\perp}$, then such a $D$ is unique.
\end{remark}
The above observations are summarized in the following theorem.
 
\begin{theorem}\label{reverse}Let $\mathcal{C}$ be an $[n,k,d]$ binary linear code with a generator matrix $\mathcal{G}$. Let $i\in[k]$. Then the following results hold:
    \begin{enumerate}[(i)]
    \item For any $R_p\in\mathcal{R}_i(\mathcal{G}_p)$, $|R_p\cap[n+1,2n]|$ must be even.
    Consequently, no recovery set $R_p\subseteq[n+1,2n]$ with $|R_p|$ odd belongs to $\mathcal{R}_i(\mathcal{G}_p)$.

    \item  Every recovery set $R_p\in\mathcal{R}_i(\mathcal{G}_p)$ induces a
recovery set in $\mathcal{R}_i(\mathcal{G})$ through the mapping
\[
R_p\longmapsto
\left(R_p\cap[n]\right)
\triangle
\{j-n:j\in R_p\cap[n+1,2n]\}.
\]
\end{enumerate}
\end{theorem}

Theorem \ref{reverse} establishes the correspondence from $\mathcal{R}_i(\mathcal{G}_p)$ to $\mathcal{R}_i(\mathcal{G})$. 
 Now, we consider the reverse direction. 
In the following theorem, for a given $R\in\mathcal{R}_i(\mathcal{G})$, we characterize the recovery sets in $\mathcal{R}_i(\mathcal{G}_p)$ induced by $R$, including the additional recovery sets arising from linear dependencies among the corresponding columns.


\begin{theorem}\label{thmR1}
    Let $\mathcal{C}$ be an $[n,k,d]$ binary linear code with a generator matrix $\mathcal{G}$. Then the following results hold:
    \begin{enumerate}[(i)]
        \item\label{thmr1} Let $R\in\mathcal{R}_i(\mathcal{G})$, where $i\in[k]$. For a nonempty subset $T\subseteq R$ with $| T|$ even, we have a unique recovery set \[R\setminus T\cup (T+n)=\widetilde{R}_T\in\mathcal{R}_i(\mathcal{G}_p),\] such that $|R|=|\widetilde{R}_T|.$
        \item\label{thmr2} For each $j\in[n]$, we have $\{j,j+n\}\in\mathcal{R}_{k+1}(\mathcal{G}_p)$. Consequently, there are $n$ disjoint recovery sets for $x_{k+1}$ of cardinality $2.$
        \item\label{thmr3} Let $i\in[k]$, let $R\in\mathcal{R}_i(\mathcal{G})$, and  there exists a nonempty subset $D=\operatorname{supp}(\mathbf{h})$ for some $\mathbf{h}\in\mathcal{C}^\perp$
with $R\cap D=\emptyset$. Assume that $D$ is the only nonempty dual support contained in
$R\cup D$. For $F_R\subseteq R$ and $F_{D}\subseteq D$ with $|F_R|$ and $|F_D|$ odd, define
\[
\mathcal{F}
=
\bigl(R\setminus F_R\bigr)\cup\bigl(D\setminus F_{D}\bigr)\cup\bigl((F_R\cup F_{D})+n\bigr).
\]
 Then $\mathcal{F}\in\mathcal{R}_i(\mathcal{G}_p)$. 


    \end{enumerate}
\end{theorem}
\begin{proof} 
  Given $R\in\mathcal{R}_i(\mathcal{G})$ and a subset $T\subseteq R$, we have \[x_i=\sum_{j\in R}c_j=\sum_{j\in R\setminus T}c_j+\sum_{j'\in T }(x_{k+1}+c_{n+j'})=\sum_{j\in R\setminus T}c_j+\sum_{j'\in T}c_{n+j'}+\lvert T\rvert x_{k+1}.\] 

   Since $|T|$ is even, the sum can be written as \[x_i=\sum_{j\in R\setminus T}c_j+\sum_{j'\in T}c_{n+j'}.\] It follows that $\widetilde{R}_T$ recovers $x_i$. It is easy to verify that $\widetilde{R}_T$ is a minimal recovery set. Hence $\widetilde{R}_T\in\mathcal{R}_i(\mathcal{G}_p).$
   
 Further, a set of coordinate positions $\{j, j+n\}$ corresponds to the codeword symbols $(c_j,c_{j+n}=c_j+x_{k+1})$. Hence, we have $\{j, j+n\}\in\mathcal{R}_{k+1}(\mathcal{G}_p)$, for all $j\in[n]$. This implies that there exist $n$ disjoint recovery sets of the form $\{j,j+n\}$ for $x_{k+1}.$

  For the last part, note that \[\sum_{j\in\mathcal{F}}c_j=\sum_{j\in R\setminus F_R}c_j+\sum_{j\in D\setminus F_D}c_j+\sum_{j\in F_R\cup F_D}c_{j+n}=\sum_{j\in R}c_j+\sum_{j\in D}c_j+(|F_R|+|F_D|)x_{k+1}=x_i.\]
  This indicates that $\mathcal{F}$ recovers $x_i$. To prove minimality, suppose that $M=A\cup (B+n)\subsetneq \mathcal{F}$, such that $M$ recovers $x_i$, $A\subseteq \bigl(R\setminus F_R\bigr)\cup\bigl(D\setminus F_{D}\bigr)$, and $B\subseteq ( F_R\cup F_D)$ with $(A,B)\neq \bigl((R\setminus F_R)\cup(D\setminus F_{D}), F_R\cup F_D\bigr)$. This implies that $(R\cup D)\setminus(A\cup B)=N$ (say), is the support of a codeword in $\mathcal C^\perp$ contained in $R\cup D.$ According to the assumption, $N=D$ or $\emptyset.$  Now, if $N=D$, then $|B|=|F_R\cup F_D|-|F_D|$ is odd, which is a contradiction. Hence, $N=\emptyset$. It follows that $A= \bigl(R\setminus F_R\bigr)\cup\bigl(D\setminus F_{D}\bigr)$, and $B=( F_R\cup F_D)$. Therefore, $\mathcal{F}\in\mathcal{R}_i(\mathcal{G}_p).$ 
  The proof concludes.
\end{proof}
\begin{corollary}\label{colR11}
     Consider $T=R$ with $|R|$ even. Then $\widetilde{R}_R\in\mathcal{R}_i(\mathcal{G}_p)$, such that $|\widetilde{R}_R|=|R|$, and $R\cap \widetilde{R}_R=\emptyset$.
\end{corollary}
\begin{corollary}\label{colR12}
    If there are $t$ disjoint recovery sets of even cardinalities in $\mathcal{R}_i(\mathcal{G)}$ for some $i\in[k]$, then there are at least $2t$ disjoint recovery sets in $\mathcal{R}_i(\mathcal{G}_p).$
\end{corollary}
\begin{corollary}\label{colR13}
    Let $i\in[k]$ and let $R\in\mathcal{R}_i(\mathcal{G})$ with $|R|$ odd. Suppose there exists a nonempty dual support
$D$ such that $R\cap D=\emptyset$, $|D|$ is odd, and no proper subset of $R\cup D$ of even
cardinality recovers $x_i$. Then $(R\cup D)+n\in\mathcal{R}_i(\mathcal{G}_p)$.
\end{corollary}



\begin{lemma}\label{lemR1}
   Let \(R\subseteq[n]\) be a set with \(|R|=r\geq 1\). Then  the number of nonempty subsets of $R$ with even-cardinality is $2^{r-1}-1$,
whereas the number of subsets of $R$ with odd-cardinality is
 $2^{r-1}$.
\end{lemma}
\begin{theorem}\label{thmR2}
    Let $\mathcal{C}$ be an $[n,k,d]$ binary linear code with a generator matrix $\mathcal{G}$. Let $R\in\mathcal{R}_i(\mathcal{G})$, where $i\in[k]$ with $|R|=r$. Then the number of new recovery sets in $\mathcal{R}_i(\mathcal{G}_p)$ obtained as in Theorem \ref{thmR1} by $R$ is $2^{r-1}-1.$
\end{theorem}
\begin{proof}
    This theorem is an immediate consequence of Part \eqref{thmr1} of Theorem \ref{thmR1} and Lemma \ref{lemR1}
\end{proof}



\noindent\textbf{Counting occurrence of the nodes $\boldsymbol{j}$ and $\boldsymbol{j+n}$.}
A node $j\in R \subset [n]$ appears in $\widetilde{R}_T$ if and only if $j\notin T$ and $j+n$ appears in $\widetilde{R}_T$ if and only if $j\in T.$ Let $|R|=r$. Since the number of possible choices for $T$ with  $|T|$ even such that $j\notin T$ is $2^{r-2}-1$, there are $2^{r-2}-1$ recovery sets $\widetilde{R}_T\in\mathcal{R}_i(\mathcal{G}_p)$ containing $j.$ The number of possible choices for $T$ with  $|T|$ even such that $j\in T$ is $2^{r-2}.$ Equivalently, this is the number of recovery sets of form $\widetilde{R}_T\in\mathcal{R}_i(\mathcal{G}_p)$ which contains $j+n.$




   We now extend the discussion to a systematic generator matrix of $\mathcal{C}$. Our objective is to determine whether additional recovery sets arise from the systematic structure.
   Since systematic generator matrices possess a column for each $i\in[k]$, this analysis provides a characterization of the recovery sets involving the node $i+n$. 
\begin{theorem}\label{thmR4}
    Let $\mathcal{C}$ be an $[n,k,d]$ binary linear code with a systematic generator matrix $\mathcal{G}$. Let $i\in[k]$, and let $j\in[n]\setminus\{i\}$. Then, for each $j$, there exists a recovery set $\widetilde{R}_j\in\mathcal{R}_i(\mathcal{G}_p)$, where $\widetilde{R}_j=\{j,j+n,i+n\}.$ 

\end{theorem}
   
  \begin{proof}
     Fix an information symbol $x_i,$ where $i\in[k].$ Given $j\in[n]\setminus\{i\}$, we have $$\sum_{l\in\widetilde{R}_j}c_l=c_j+c_{j+n}+c_{i+n}=c_j+c_j+x_{k+1}+c_i+x_{k+1}=c_i=x_i.$$ 
     
     This implies that $\widetilde{R}_j\in\mathcal{R}_i(\mathcal{G}_p).$ 
  \end{proof}

\begin{corollary}\label{colR41}
For a systematic generator matrix $\mathcal{G}$ of $\mathcal{C}$, if there exists $t$ recovery sets of cardinality $3$ in $\mathcal{R}_i(\mathcal{G})$ such that a node $j\in[n]$ lies in $t_j$ of these recovery sets, then there are at least $3t+n-1$ new recovery sets of cardinality $3$ in $\mathcal{R}_i(\mathcal{G}_p)$, where $j$ lies in $t_j+1$ of these recovery sets and $j+n$ lies in $2t_j+1$ of these recovery sets.
\end{corollary}


\begin{remark}
   Given a systematic generator matrix $\mathcal{G}$ of $\mathcal{C}$, each odd cardinality subset $J\subset[n]$ gives rise to a recovery set of the form $J\cup J+n\cup \{i+n\}$ for $x_i$ corresponding to $\mathcal{G}_p$, but for $|J|\geq3$ it is not a minimal recovery set.
\end{remark}
The following theorems provide methods to obtain recovery sets in $\mathcal{R}_i(\mathcal{G}_p)$ (recall that recovery sets in $\mathcal{R}_i(\mathcal{G}_p)$ are minimal) using an odd cardinality subset of minimal inclusion dual support.
\begin{theorem}\label{thmR6}
    Let $\mathcal{G}$ be a systematic generator matrix of $\mathcal{C}$. Let $i\in[k]$, and let $\mathbf{h}\in\mathcal{C}^\perp$ such that $D=\operatorname{supp}(\mathbf{h})$ is a minimal-inclusion dual support and $i\notin D$. Suppose that there is no element $R\in\mathcal{R}_i(\mathcal{G})$ satisfying $R\subsetneq D.$ For $J\subseteq D$ with $|J|$ odd, define
\[
\mathcal{D}_J
=
\big(D\setminus J\big)
\cup(J+n)
\cup\{i+n\}.
\]
Then $\mathcal{D}_J\in\mathcal{R}_i(\mathcal{G}_p)$. Consequently, there are $2^{|D|-1}$ such distinct recovery sets in $\mathcal{R}_i(\mathcal{G}_p)$, each containing $i+n$.
\end{theorem}
\begin{proof}
    
  Since $|J|$ is odd, we have
\[\sum_{j\in \mathcal{D}_J}c_j=\Big(\sum_{j\in D\setminus J}c_j\Big)+\Big(\sum_{j\in J}c_{j+n}\Big)+c_{i+n}=
\sum_{j\in D}c_j+c_i=c_i.\]
Consequently, $\mathcal{D}_J$ is a recovery set for the $i$-th information symbol.

Next, we show that $\mathcal{D}_J$ is a minimal recovery set. On the contrary, suppose that there exists $T\subsetneq \mathcal{D}_J$ such that $T\in \mathcal{R}_i(\mathcal{G}_p).$ Then, $T$ can be written as \[T=A \cup (B+n)\cup \big(\epsilon \{i+n\}\big),\] where $A\subseteq D\setminus J,  B\subseteq J$, and $\epsilon\in\{0,1\}.$ We consider two cases: (i) $ \epsilon=0$, and (ii) $\epsilon=1.$

\begin{enumerate}
    \item[(i)] 
Let $\epsilon=0.$ Then $|B|$ is even and 
 we have \[\sum_{j\in T}c_j=\sum_{j\in A }c_j+\sum_{j\in B }c_j=c_i,\] which implies that $A\cup B \in \mathcal{R}_i(\mathcal{G})$. But $A\cup B\subsetneq D$, which is a contradiction. Hence, $\mathcal{D}_J\in \mathcal{R}_i(\mathcal{G}_p)$.

\item[(ii)] Let $\epsilon=1.$
Then $|B|$ is odd and $$\sum_{j\in A \cup B}c_j=0.$$ Furthermore, $ A \cup B\subsetneq D,$ and since $\{i+n\}$ alone does not recover $x_i$, we have $ A \cup B\neq \emptyset$. By assumption, $D$ is a minimal-inclusion dual support, and hence, we obtain a contradiction. Therefore, in this case as well, $\mathcal{D}_J\in \mathcal{R}_i(\mathcal{G}_p)$.\end{enumerate} 
\end{proof}

\begin{theorem}\label{thmR7}
Let $\mathcal{C}$ be a binary linear $[n,k,d]$ code with a systematic generator matrix $\mathcal{G}$. For $i\in[k]$, let $D=\operatorname{supp}(\mathbf h),$ be an inclusion-minimal support of a nonzero codeword of $\mathbf h\in \mathcal{C}^{\perp}$, and $i\in D$.  Let $\ell\in[n]\setminus D$ such that that there does not exists a recovery set $R$ for $x_i$ satisfying 
$\ell\in R\subsetneq D\cup\{\ell\}$.

For $E\subseteq D\setminus\{i\}$ with $|E|$ odd, define
\[
\mathcal{D}_E
=
\bigl(D\setminus(E\cup\{i\})\bigr)
\cup(E+n)
\cup\{\ell,\ell+n\}.
\]
Then $\mathcal{D}_E\in\mathcal{R}_i(\mathcal{G}_p)$.
\end{theorem}

\begin{proof}
Since $\mathbf h\in\mathcal{C}^{\perp}$ 
$$\sum_{j\in\mathcal{D}_E}c_j=
\sum_{j\in D\setminus(E\cup\{i\})}
c_j
+
\sum_{j\in E}
(c_j+x_{k+1})
+c_\ell+c_\ell+x_{k+1}=
c_i.
$$
 Thus, $\mathcal{D}_E$ is a recovery set for the $i$-th information symbol.

It remains to prove that $\mathcal{D}_E$ is minimal. Suppose, to the contrary, that there exists a proper subset $
T\subsetneq\mathcal{D}_E$
such that $T\in\mathcal{R}_i(\mathcal{G}_p)$. We may write
$T=A\cup(B+n)\cup F$, where $A\subseteq D\setminus(E \cup\{i\}), B\subseteq E,
F\subseteq\{\ell,\ell+n\}.$
We consider the possible choices of $F$.

\medskip
\begin{enumerate}[(i)]
    \item $F=\emptyset$ or $F=\{\ell,\ell+n\}$. Then 
\[
\sum_{j\in T}c_j=\sum_{j\in A\cup B}c_j=c_i.
\]
Thus, if
$A\cup B\cup\{i\}\subsetneq D$,
we obtain a contradiction to the inclusion-minimality of $D$ as the support of a codeword of dual. It remains to consider the case
$A\cup B\cup\{i\}=D$.
In this case, we have $B=E$. If $F=\emptyset$, then $|B|$ is even, contradicting the assumption that $|E|$ is odd. On the other hand, if $F=\{\ell,\ell+n\}$, then $|B|$ is odd, which implies that $T=\mathcal{D}_E$, again contradicting the assumption that $T\subsetneq\mathcal{D}_E$.
\item   $F=\{\ell\}$ or $F=\{\ell+n\}$. In this case, we have
\[
\sum_{j\in T}c_j=\sum_{j\in A\cup B}c_j=c_i+c_\ell.
\]
Consequently,
 $(A\cup B)\cup\{\ell\}$ is a recovery set for $x_i$ in $\mathcal{C}$ satisfying $\ell\in(A\cup B)\cup\{\ell\}
\subseteq D\cup\{\ell\}$,
contradicting the assumption on $\ell$.
\end{enumerate}
Therefore,
$\mathcal{D}_E \in\mathcal{R}_i(\mathcal{G}_p)$,
which completes the proof.
\end{proof}
Since there are $2^{|D|-2}$ odd cardinality subsets $E$ of $D$, for each $\ell\in[n]\setminus D$, satisfying the condition of Theorem \ref{thmR7}, there exist $2^{|D|-2}$ recovery sets $D_E$. 
The idea of Theorem \ref{thmR7} can be extended by replacing the element $\ell$ with a subset of odd cardinality to define \(\mathcal{D}_E\). This set still recovers $x_i$ but is not a minimal recovery set.

Now we introduce the following lemma, which plays a crucial role in proving Theorem \ref{thmR8}.

\begin{lemma}\label{nosupport}
Let $i\in[k]$ and $R_p\in\mathcal{R}_i(\mathcal{G}_p)$ such that $R_p^{(1)}\cap R_p^{(2)}\neq \emptyset$. Then $R_p^{(1)}\triangle R_p^{(2)}$ contains no nonempty dual support, i.e.,
$R_p^{(1)}\triangle R_p^{(2)}\in\mathcal{R}_i(\mathcal{G}).$
\end{lemma}
\begin{proof}
    On the contrary, suppose that $D\subseteq R_p^{(1)}\triangle R_p^{(2)}$ is a nonempty dual support. Note that $$R_p^{(2)}=\bigl(R_p^{(2)}\setminus(R_p^{(1)}\cup D)\bigr)\cup (R_p^{(2)}\cap R_p^{(1)}) \cup (R_p^{(2)}\cap D).$$ Since $|R_p^{(2)}|$ is even, by Remarks \ref{cap12}  and \ref{2D}, the cardinality of the set $\bigl(R_p^{(2)}\setminus(R_p^{(1)}\cup D)\bigr)$ is even.

We already know that $\bigl(R_p^{(1)}\triangle R_p^{(2)}\bigr)\setminus D$ recovers $x_i.$ Moreover, \[\bigl(R_p^{(1)}\triangle R_p^{(2)}\bigr)\setminus D=\bigl(R_p^{(1)}\setminus(R_p^{(2)}\cup D)\bigr)\cup \bigl(R_p^{(2)}\setminus(R_p^{(1)}\cup D)\bigr).\] It follows that $\bigl(R_p^{(1)}\setminus(R_p^{(2)}\cup D)\bigr)\cup \bigl(\big(R_p^{(2)}\setminus(R_p^{(1)}\cup D)\big)+n\bigr)$ also recovers $x_i$, and is properly contained in $R_p.$ This contradicts the minimality of $R_p.$ Hence the proof.
    
\end{proof}
The following theorem provides a systematic characterization of all minimal recovery sets for $i\in[k]$ in $\mathcal{G}_p$, i.e., $\mathcal{R}_i(\mathcal{G}_p)$, in terms of the recovery sets associated with the matrix $\mathcal G$, namely $\mathcal{R}_i(\mathcal G)$. In particular, it establishes a direct correspondence between the recovery set structure of $\mathcal C$ and  $\mathcal{C}_p$, thereby giving a complete method for determining $\mathcal{R}_i(\mathcal{G}_p)$ from $\mathcal{R}_i(\mathcal G)$.

\begin{theorem}\label{thmR8}
Let $\mathcal{C}$ be an $[n,k,d]$ binary linear code with a systematic generator matrix $\mathcal{G}$, and let $i\in[k]$. For each $R\in\mathcal{R}_i(\mathcal{G})$, let $\mathscr{K}_R$ denote the collection of recovery sets obtained from $R$ by Theorem~\ref{thmR1}. Let $\mathscr{L}$, $\mathscr{M}$, and $\mathscr{N}$ denote the collections of recovery sets obtained by Theorems~\ref{thmR4}, \ref{thmR6}, and \ref{thmR7}, respectively. Then
\[
\mathcal{R}_i(\mathcal{G}_p)
=
\mathcal{R}_i(\mathcal{G})
\cup \mathscr{L}\cup \mathscr{M}\cup \mathscr{N}
\cup
\bigcup_{R\in\mathcal{R}_i(\mathcal{G})}\mathscr{K}_R.
\]
\end{theorem}

\begin{proof}
It has already been established that
\[
\mathcal{R}_i(\mathcal{G})
\cup \mathscr{L}\cup \mathscr{M}\cup \mathscr{N}
\cup
\bigcup_{R\in\mathcal{R}_i(\mathcal{G})}\mathscr{K}_R
\subseteq
\mathcal{R}_i(\mathcal{G}_p).
\]
It remains to prove the reverse inclusion. Let
$R_p\in\mathcal{R}_i(\mathcal{G}_p)$, and define
\[
R_p^{(1)}=R_p\cap[n],
\qquad
R_p^{(2)}
=
\{j-n:j\in R_p\cap[n+1,2n]\}.
\]
Then $R_p^{(1)}\triangle R_p^{(2)}$ is a recovery set for $x_i$, and $|R_p^{(2)}|$ is even. Moreover, if there exists
$\mathbf{h}\in\mathcal{C}^{\perp}$ such that
\[
D=\operatorname{supp}(\mathbf{h})
\subsetneq R_p^{(1)}\triangle R_p^{(2)},
\]
then $\bigl(R_p^{(1)}\triangle R_p^{(2)}\bigr)\setminus D
\in\mathcal{R}_i(\mathcal{G}).$ 

The argument splits into two parts according to whether $R_p^{(1)}=\emptyset$. We treat the two possibilities separately.

First, suppose that $R_p^{(1)}=\emptyset$. Then $R_p=R_p^{(2)}+n.$ For this case, we distinguish two further subcases.

\begin{enumerate}[(I)]
    \item There is no set $D\subsetneq R_p^{(2)}$ such that
$D=\operatorname{supp}(\mathbf{h})$ for some
$\mathbf{h}\in\mathcal{C}^{\perp}$. Then
$R_p^{(2)}\in\mathcal{R}_i(\mathcal{G}).$
By Corollary~\ref{colR11} of Theorem~\ref{thmR1},
\[
R_p^{(2)}+n
=
R_p
\in
\mathscr{K}_{R_p^{(2)}}.
\]
\item
There exists
$D\subsetneq R_p^{(2)}$
such that
$D=\operatorname{supp}(\mathbf{h})$ for some
$\mathbf{h}\in\mathcal{C}^{\perp}$. Then
$R_p^{(2)}\setminus D\in\mathcal{R}_i(\mathcal{G})$,
and, since $|R_p^{(2)}|$ is even and $|D|$ is odd, $|R_p^{(2)}\setminus D|$
is odd. Two subcases arise.

\begin{enumerate}[(i)]
    \item Suppose
$
R_p^{(2)}\setminus D=\{i\}.
$
Then
$
R_p^{(2)}=D\cup\{i\}
$. Taking $J=D$ in Theorem~\ref{thmR6}, we obtain
\[
\mathcal{D}_J
=
(D+n)\cup\{i+n\}
=
R_p^{(2)}+n
=
R_p
\in\mathscr{M}.
\]

\item Suppose
$
R_p^{(2)}\setminus D\neq\{i\}.
$
Since both $|R_p^{(2)}\setminus D|$ and $|D|$ are odd, Corollary~\ref{colR13} yields
\[
\bigl((R_p^{(2)}\setminus D)\cup D\bigr)+n
=
R_p^{(2)}+n
=
R_p
\in
\mathscr{K}_{R_p^{(2)}\setminus D}.
\]
\end{enumerate}
\end{enumerate}
Now suppose that $R_p^{(1)}\neq\emptyset$.
In this part as well, we distinguish two subcases.

\begin{enumerate}[(I)]
    \item There is no set
$
D\subsetneq R_p^{(1)}\triangle R_p^{(2)}
$
such that
$D=\operatorname{supp}(\mathbf{h})$ for some
$\mathbf{h}\in\mathcal{C}^{\perp}$. This implies that
\[
R_p^{(1)}\triangle R_p^{(2)}
\in\mathcal{R}_i(\mathcal{G}).
\]

If
$
R_p^{(1)}\cap R_p^{(2)}=\emptyset,
$
then
$
R_p^{(1)}\triangle R_p^{(2)}
=
R_p^{(1)}\cup R_p^{(2)}
\in\mathcal{R}_i(\mathcal{G}).
$
Since $|R_p^{(2)}|$ is even, Part~\eqref{thmr1} of Theorem~\ref{thmR1} gives
\[
R_p^{(1)}\cup(R_p^{(2)}+n)
=
R_p
=
\widetilde{R}_{R_p^{(2)}}
\in\mathscr{K}_{R_p^{(2)}}.
\]

Next, suppose that
$
R_p^{(1)}\cap R_p^{(2)}=\{m\}.
$
There are again two possibilities.

\begin{enumerate}[(i)]
    \item Suppose
$
R_p^{(1)}\triangle R_p^{(2)}=\{i\}.
$
Then
$
R_p^{(1)}\setminus R_p^{(2)}=\emptyset,$ and $
R_p^{(2)}\setminus R_p^{(1)}=\{i\}.
$
Consequently,
\[
R_p^{(2)}
=
R_p^{(1)}\cup\{i\}
=
\{m,i\}.
\]
Therefore,
\[
R_p
=
R_p^{(1)}\cup(R_p^{(2)}+n)
=
\{m,m+n,i+n\}
\in\mathscr{L}.
\]

\item  Suppose
$
R_p^{(1)}\triangle R_p^{(2)}\neq\{i\}.
$
Then
$
D'
=
\bigl(R_p^{(1)}\triangle R_p^{(2)}\bigr)\cup\{i\}
$
is the support of a codeword of $\mathcal{C}^{\perp}$. Note that
$
i\in D'
\text{ and }
m\notin D'.
$
Since $|R_p^{(2)}\setminus R_p^{(1)}|$ is odd, taking
\[
E=R_p^{(2)}\setminus R_p^{(1)}
\qquad\text{and}\qquad
\ell=m
\]
in Theorem~\ref{thmR7} gives
\[
\begin{aligned}
\mathcal{D}_E
&=
\bigl(R_p^{(1)}\setminus R_p^{(2)}\bigr)
\cup
\bigl((R_p^{(2)}\setminus R_p^{(1)})+n\bigr)
\cup\{m,m+n\}\\
&=
R_p^{(1)}\cup(R_p^{(2)}+n)
=
R_p
\in\mathscr{N}.
\end{aligned}
\]

\end{enumerate} 

\item
There exists
$
D\subsetneq R_p^{(1)}\triangle R_p^{(2)}
$
such that
$D=\operatorname{supp}(\mathbf{h})$ for some
$\mathbf{h}\in\mathcal{C}^{\perp}$. By Lemma~\ref{nosupport},
$
R_p^{(1)}\cap R_p^{(2)}=\emptyset.
$
Consequently,
\[
\bigl(R_p^{(1)}\triangle R_p^{(2)}\bigr)\setminus D
=
(R_p^{(1)}\setminus D)\cup(R_p^{(2)}\setminus D)
\in\mathcal{R}_i(\mathcal{G}).
\]
Now consider two subcases.

\begin{enumerate}[(i)]
    \item Suppose
$
(R_p^{(1)}\setminus D)\cup(R_p^{(2)}\setminus D)=\{i\}.
$
Since $i\notin R_p^{(1)}$, we have
$
R_p^{(1)}\setminus D=\emptyset,
$ and $
R_p^{(2)}\setminus D=\{i\}.
$
Hence
\[
R_p^{(1)}\subseteq D
\qquad\text{and}\qquad
R_p^{(2)}=(R_p^{(2)}\cap D)\cup\{i\}.
\]
Let
$
J=R_p^{(2)}\cap D.
$
Since $|R_p^{(2)}|$ is even, $|J|$ is odd. Applying Theorem~\ref{thmR6}, we obtain
\[
\begin{aligned}
\mathcal{D}_J
&=
R_p^{(1)}
\cup
\bigl((R_p^{(2)}\cap D)+n\bigr)
\cup\{i+n\}\\
&=
R_p^{(1)}\cup(R_p^{(2)}+n)
=
R_p
\in\mathscr{M}.
\end{aligned}
\]

\item Suppose
$
(R_p^{(1)}\setminus D)\cup(R_p^{(2)}\setminus D)
\neq\{i\}.
$
Let
\[
F_D=R_p^{(2)}\cap D,
\qquad
F_R=R_p^{(2)}\setminus D.
\]
Since $|R_p^{(2)}|$ is even, $|F_D|$ and $|F_R|$ have the same parity. By the preceding result, both are odd. Applying Part~\eqref{thmr3} of Theorem~\ref{thmR1}, we obtain
\[
\begin{aligned}
\mathcal{F}
&=
(R_p^{(1)}\setminus D)
\cup(R_p^{(1)}\cap D)\cup
\bigl((R_p^{(2)}\setminus D)
\cup(R_p^{(2)}\cap D)\bigr)+n\\
&=
R_p^{(1)}\cup(R_p^{(2)}+n)
=
R_p\in
\mathscr{K}_{(R_p^{(1)}\setminus D)\cup
(R_p^{(2)}\setminus D)}.
\end{aligned}
\]

\end{enumerate} 

\end{enumerate} 
This completes the proof.
\end{proof}

\begin{remark}
    It is important to note that, in Theorems \ref{thmR4} to \ref{thmR8}, systematicity is used only to establish the single fact $c_i = x_i$ for the particular index $i$ under consideration. The hypothesis that $\mathcal{G}$ is a systematic generator matrix is therefore stronger than necessary: the theorems also hold if $\mathcal G$ has a systematic column for $x_i$.
\end{remark}

\begin{example}\label{exR1}
 Consider a $[6,3]$ binary linear code $\mathcal{C}$ with a generator matrix as \[\mathcal{G}=\begin{pmatrix}
        1 &0&0&1&0&1\\
        0&1&0&1&1&0\\
        0&0&1&0&1&1\\
    \end{pmatrix}.\]     The minimal recovery sets of $x_1, x_2$, and $x_3$ are $$\mathcal{R}_1(\mathcal{G})=\{(1),(2,4),(3,6),(2, 5, 6),(3,4,5)\},$$
    \[
\mathcal{R}_2(\mathcal{G})=\{(2),(1,4),(3,5),(1, 5, 6),(3, 4, 6)\}, \text{ and }\] \[\mathcal{R}_3(\mathcal{G})=\{(3),(1, 6), (2, 5), (1, 4, 5),(2, 4, 6)\},
\]
respectively. Now consider \(\mathcal{C}_p\).
By Part \eqref{thmr2} of Theorem \ref{thmR1}, we have 
\[\{\{1,7\},\{2,8\},\{3,9\},\{4,10\},\{5,11\},\{6,12\}
\}\subset\mathcal{R}_4(\mathcal{G}_p).\]

For the information symbol $x_1$, we obtain a set $\mathscr K$ by Theorem \ref{thmR1} (the first two sets by Corollary \ref{colR11}), a set $\mathscr L$ by Theorem \ref{thmR4}, a set $\mathscr M$ by Theorem \ref{thmR6}, and a set $\mathscr N$ by Theorem \ref{thmR7} as follows:
\[\begin{aligned}
\mathscr K=\{&
\{8,10\},\{9,12\},\{2,11,12\},\{5,8,12\},\{6,8,11\},\{3,10,11\},\{4,9,11\},\{5,9,10\}\}
\\
\mathscr L=\{&\{2,7,8\},
\{3,7,9\},\{4,7,10\},\{5,7,11\},
\{6,7,12\}
\}\\
\mathscr M=\{&\{2,3,7,11\},\{2,5,7,9\},\{3,5,7,8\},\{7,8,9,11\},\{4,5,7,12\},\{4,6,7,11\},\{5,6,7,10\},\{7,10,11,12\}\},\\
\mathscr N=\{&\{2,3,8,12\},\{2,3,9,10\},
\{4,6,8,12\},\{2,5,10,11\},\{2,6,8,9\},\{2,6,10,12\},
\{3,4,8,9\},\{3,4,10,12\},\\&\{3,5,11,12\},\{4,6,9,10\},\{4,5,8,11\},\{5,6,9,11\}\}
\end{aligned}.
\]

Additionally, we have
$\mathcal{R}_1(\mathcal{G}_p)=\mathcal{R}_1(\mathcal{G})\cup \mathscr K\cup \mathscr L\cup \mathscr M\cup \mathscr N.$ The same conclusion holds for $i=2$ and $i=3,$ i.e., $\mathcal{R}_i(\mathcal{G}_p)$ can be obtained by the same method for $i=2,$ and $3$.
\end{example}

\subsection{Parameters of the Iterated Plotkin Construction}\label{plotkincpm}

We now determine the parameters of the codes obtained by recursively applying the Plotkin construction. Let
$\mathcal{C}_{p^0}=\mathcal{C}, \text{ and } \mathcal{G}_{p^0}=\mathcal{G}$,
where $\mathcal{C}$ is an $[n,k,d]$ binary linear code with a generator matrix $\mathcal{G}$. For $m\geq 1$, define
\begin{equation}\label{Cp^m}
    \mathcal{C}_{p^{m}}=\{(\mathbf{c},\mathbf{c})\mid \mathbf c\in \mathcal{C}_{p^{m-1}}\}\cup\{(\mathbf{c},\mathbf{1}_n+\mathbf{c})\mid \mathbf c\in \mathcal{C}_{p^{m-1}}\},
\end{equation} and \[\mathcal{G}_{p^{m}}=\begin{pmatrix}
    \mathcal{G}_{p^{m-1}} &\mathcal{G}_{p^{m-1}}\\\mathbf{0}_n &\mathbf{1}_n
\end{pmatrix}.\]
Since the Plotkin construction maps an $[n,k,d]$ code to an $[2n,k+1,\min\{n,2d\}]$ code, it can be verified by induction on $m$ that  $\mathcal{C}_{p^m}$ is an $\left[2^{m} n,\;k+m,\;\min\left\{2^{m-1}n,\;2^{m} d\right\}\right]$ linear code with  a generator matrix $\mathcal{G}_{p^{m}}$.





\begin{theorem}\label{thmI1}
   Let $\mathcal{C}$ be an $[n,k,d]$ linear code over $\mathbb{F}_2$ and let $\mathcal{G}$ be a generator matrix of $\mathcal{C}$. Suppose that for some $i\in[k]$, there exists a systematic node in $\mathcal{G}$ and  $t$ recovery sets of cardinality $3$ in $\mathcal{R}_i(\mathcal{G})$ such that $j\in[n]$  is in $t_j$ of these cardinality $3$ recovery sets. Let $m$ be a positive integer and consider a code $\mathcal C_{p^m}$, iteratively obtained from $\mathcal{C}$ (as in \eqref{Cp^m}). Then $\mathcal C_{p^m}$ has at least
$4^m t + 2n4^{m-1}-n2^{m-1} -\frac{4^m-1}{3}$ recovery sets of cardinality $3$ in $\mathcal{R}_i(\mathcal{G}_{p^{m}})$ such that the nodes $j$ and $j+n$ is contained in $2^mt_j+2^m-1$ of these recovery sets.
\end{theorem}

\section{Allocations for service rates in \texorpdfstring{$\mathcal{C}_p$}{Cp} and \texorpdfstring{$\mathcal{C}_{p^m}$}{Cp^m}}\label{sec4}
In this section, we determine the maximal achievable service rates in service polytope $\Lambda(\mathcal{G}_p)$, in terms of the maximal achievable service rates in service polytope $\Lambda(\mathcal{G})$. Let $\boldsymbol{\hat{\lambda}}=(\hat{\lambda}_1, \hat{\lambda}_2,\ldots,\hat{\lambda}_k, \hat{\lambda}_{k+1})$ denote the achievable service request vector in $\Lambda(\mathcal{G}_p)$. Since $\mathcal{R}_i(\mathcal{G})\subseteq\mathcal{R}_i(\mathcal{G}_p),$ for every $i\in[k]$, we have $\Lambda(\mathcal{G})\subseteq\Lambda(\mathcal{G}_p)$.

Now, we first analyze the partial hypergraph $\Gamma_{\mathcal{G}_p}^{k+1}$, for the additional information symbol ${k+1}$ present in $\mathcal{C}_p$, i.e., $I=\{k+1\}$.
We construct explicit matching and vertex cover for $\Gamma_{\mathcal{G}_p}^{k+1}$. By Part \eqref{thmr2} of Theorem \ref{thmR1}, there are exactly $n$ disjoint hyperedges (recovery sets) of size $2$ labeled $\mathbf e_{k+1}$ in $\Gamma_{\mathcal{G}_p}^{k+1}$. In other words, $n\leq \nu(\Gamma_{\mathcal{G}_p}^{{k+1}})$. Now consider the following assignment:
$$\beta_v=\begin{cases}
   1 \quad;&  \text{if }v\in[n+1, 2n] \\0 \quad ;& \text{if } v\in[n].
\end{cases} $$
Since the first $n$ columns of $\mathcal G_p$ all have last coordinate zero, $\mathbf e_{k+1}$ is not in their span. It follows that any $R\in \mathcal{R}_{k+1}(\mathcal{G}_p)$ contains at least a node (vertex) $v\in[n+1,2n]$, therefore, the set $[n+1,2n]$ is a valid vertex cover for $\Gamma_{\mathcal{G}_p}^{{k+1}}$. It is easy to confirm that $\tau(\Gamma_{\mathcal{G}_p}^{{k+1}})=n.$ Hence, by \eqref{inequalities}, 
\begin{equation}
\label{new}n=\nu(\Gamma_{\mathcal{G}_p}^{k+1})=\hat{\lambda}_{k+1}^{*}=\tau(\Gamma_{\mathcal{G}_p}^{{k+1}}).\end{equation}

Let $i\in[k]$. Now, if there are $t$ disjoint recovery sets of even cardinality in $\mathcal{R}_i(\mathcal{G})$, we have $t\leq\nu(\Gamma_{\mathcal{G}}^{i})\leq\lambda_i^*.$ Then by Corollary \ref{colR12}, we have $2t\leq\nu(\Gamma_{\mathcal{G}_p}^{i})\leq\hat{\lambda}_i^*.$ If $\mathcal G$ is systematic then $2t+1\leq\hat{\lambda}_i^*.$ 

On the other hand, for the upper bound on service rate, we focus on the vertex cover number for $\Gamma_{\mathcal{G}_p}^{i}$. Let $V$ be a vertex cover for $\Gamma_{\mathcal{G}}^{i}$ such that $|V|=\tau(\Gamma_{\mathcal{G}}^{i})$. Define $V_p=V\cup (V+n).$ We show that $V_p$ forms a valid vertex cover for $\Gamma_{\mathcal{G}_p}^{i}.$ Since each recovery set in $\mathcal{R}_i(\mathcal{G}_p)$ induces
a recovery set in $\mathcal{R}_i(\mathcal{G})$ (refer to Theorem \ref{reverse}), a vertex cover of
$\Gamma_{\mathcal{G}_p}^{i}$ can be constructed by extending a vertex cover
of $\Gamma_{\mathcal{G}}^{i}$.
For any $R_p\in\mathcal{R}_i(\mathcal{G}_p)$, there exists an $R\in\mathcal{R}_i(\mathcal{G})$ such that $R\subseteq R_p^{(1)}\triangle R_p^{(2)}$. Choose $z\in R\cap V$.
Since $\big(R_p^{(1)}\cup  (R_p^{(2)}+n )\big)=R_p,$ either $z\in R_p$ or $z+n\in R_p$. In other words, the recovery set $R_p$ is covered by $V_p.$
  Therefore, \[\hat{\lambda}_i^*\leq 2|V| \quad \text{ for all } i\in[k].\]

This idea is generalized for any subset $I\subseteq[k]$, and following the similar steps, it can be verified that for a vertex cover $V$ of $\Gamma_{\mathcal{G}}(I)$, $V_p=V\cup (V+n)$ is a vertex cover of $\Gamma_{\mathcal{G}_p}(I)$. Moreover, it directly follows from the definition of matching that $\nu(\Gamma_{\mathcal{G}}(I))\leq\nu(\Gamma_{\mathcal{G}_p}(I)).$
\begin{theorem}\label{VC1}
    Let $\mathcal{C}$ be an $[n,k,d]$ binary linear code with a generator matrix $\mathcal{G}$. Let $I\subseteq[k]$, and $V$ be a optimal vertex cover for $\Gamma_{\mathcal{G}}(I)$. Then $V_p=V\cup (V+n)$ is a vertex cover of $\Gamma_{\mathcal{G}_p}(I)$, and \[\sum_{i\in I}\hat{\lambda}_i\leq \tau(\Gamma_{\mathcal{G}_p}(I))\leq 2\tau(\Gamma_\mathcal{G}(I)) =2|V|.\]
\end{theorem}

\begin{lemma}\label{}
 Let $I\subseteq[k]$, and let $\delta:[n]\to[0,1]$ be a fractional vertex cover of
$\Gamma_{\mathcal{G}}(I)$. Define $\delta_p:[2n]\to[0,1]$ by
\[
\delta_p(j)=\delta_p(j+n)=\delta(j),
\qquad j\in[n].
\]
Then $\delta_p$ is a fractional vertex cover of
$\Gamma_{\mathcal{G}_p}(I)$. Consequently, $\tau_f\bigl(\Gamma_{\mathcal{G}_p}(I)\bigr)
\leq
2\tau_f\bigl(\Gamma_{\mathcal{G}}(I)\bigr).$

\end{lemma}

\begin{proof}
Let $R_p\in\mathcal{R}_i(\mathcal{G}_p)$ be induced by
$R\in\mathcal{R}_i(\mathcal{G})$ for some $i\in I$. Since $R\subseteq R_p^{(1)}\triangle R_p^{(2)}$, for each $j\in R$,
exactly one of $j$ and $j+n$ belongs to $R_p$. Hence,
\[
\sum_{v\in R_p} \delta_p(v)
=\sum_{j\in R_p\cap[n]}
\delta_p(j)+\sum_{j+n\in R_p\cap[2n]}\delta_p(j)
\geq
\sum_{j\in R}\delta(j)
\geq 1.
\]
If there exists a systematic node for the $i$-th information object in $\mathcal{G}$, then $\delta(i)=1$. 
Then, for all $R_p\in\mathcal{R}_i(\mathcal{G}_p)$ with $i\in R_p$ or $i+n\in R_p$, we have $\sum\limits_{v\in R_p} \delta_p(v)\geq 1.$
Thus, every hyperedge in
$\Gamma_{\mathcal{G}_p}(I)$ is fractionally covered by $\delta_p$. Moreover,
\[
\sum_{v\in[2n]}\delta_p(v)
=
\sum_{j\in[n]}
\bigl(\delta_p(j)+\delta_p(j+n)\bigr)
=
2\sum_{j\in[n]}\delta(j).
\]
Taking $\delta$ to be an optimal fractional vertex cover of
$\Gamma_{\mathcal{G}}(I)$ yields $\tau_f\bigl(\Gamma_{\mathcal{G}_p}(I)\bigr)
\leq
2\tau_f\bigl(\Gamma_{\mathcal{G}}(I)\bigr)$.
\end{proof}
\begin{theorem}
    Let $\mathcal{C}$ be an $[n,k,d]$ binary linear code with a generator matrix $\mathcal{G}$. Let $I\subseteq[k]$ and let $\boldsymbol{\hat{\lambda}}=(\hat{\lambda}_1, \hat{\lambda}_2,\ldots,\hat{\lambda}_k, \hat{\lambda}_{k+1})\in \Lambda(\mathcal{G}_p)$. Then, $$\sum_{i\in I}\hat{\lambda}_i\leq\left(\sum_{i\in I}\hat\lambda_i\right)^* \le 2\left(\sum_{i\in I}\lambda_i\right)^*=2\tau_f(\Gamma_\mathcal{G}(I)).$$ Moreover, the second inequality bound is attained.
\end{theorem}
    
\begin{corollary}
    For $\boldsymbol{\lambda}\in\Lambda(\mathcal{G})$ and $\boldsymbol{\hat{\lambda}}\in\Lambda(\mathcal{G}_p)$, we have $\hat{\lambda}_i^*\leq2\lambda_i^* \text{ for all } i\in [k].$
\end{corollary}

\begin{theorem}\label{ThmR11}
Let $\mathcal{C}$ be an $[n,k,d]$ binary linear code with a generator matrix
$\mathcal{G}$. Let $\mathcal{I}\subseteq[k+1]$, and define
$\beta_p:[2n]\to\{0,1\}$ by
\[
  \beta_p(j)=\begin{cases}
    |\operatorname{supp}(\mathbf{g}_j)| \bmod 2; & \text{if } j\in[n],\\[2pt]
    1-\bigl(|\operatorname{supp}(\mathbf{g}_{j-n})| \bmod 2\bigr); & \text{if } j\in[n+1,2n].
  \end{cases}
\]
Then $\beta_p$ is a vertex cover of $\Gamma_{\mathcal{G}_p}(\mathcal{I})$.
Consequently, $\sum_{i\in\mathcal{I}}\hat{\lambda}_i\leq n$, and if
$(k+1)\in\mathcal{I}$, this bound is achievable.
\end{theorem}

\begin{proof}
Let  $\mathbf{g}_v^p$ denote the $v$-th column of $\mathcal{G}_p$, so that
$\mathbf{g}_j^p=(\mathbf{g}_j,0)^{\top}$ and
$\mathbf{g}_{j+n}^p=(\mathbf{g}_j,1)^{\top}$ for all $j\in[n]$. Then
$\operatorname{supp}(\mathbf{g}_j^p)=\operatorname{supp}(\mathbf{g}_j)$ and
$\operatorname{supp}(\mathbf{g}_{j+n}^p)=\operatorname{supp}(\mathbf{g}_j)\cup\{k+1\}$. It follows that 
\[
  \beta_p(v)=|\operatorname{supp}(\mathbf{g}^p_v)|\bmod 2
  \qquad\text{for every } v\in[2n].
\]

Let $i\in\mathcal{I}$ and let $R_p\in\mathcal{R}_i(\mathcal{G}_p)$. Then
\[
  \sum_{j\in R_p}\mathbf{g}^p_j=\mathbf{e}_i .
\]
Moreover,
\[
  \sum_{j\in R_p}|\operatorname{supp}(\mathbf{g}^p_j)|
  \;\equiv\;|\operatorname{supp}(\mathbf{e}_i)|\;=\;1 \pmod 2 .
\]
Hence
\[
  \sum_{j\in R_p}\beta_p(j)
  \;\equiv\;\sum_{j\in R_p}|\operatorname{supp}(\mathbf{g}^p_j)|\;\equiv\;1 \pmod 2.
\]
This implies that $\sum_{j\in R_p}\beta_p(j)$ is odd and in particular at least $1$. Thus every
hyperedge of $\Gamma_{\mathcal{G}_p}(\mathcal{I})$ is covered by $\beta_p$.

Moreover, for each $j\in[n]$ we have
$\beta_p(j)+\beta_p(j+n)=1$, and therefore
\[
  \sum_{v\in[2n]}\beta_p(v)=\sum_{j\in[n]}\bigl(\beta_p(j)+\beta_p(j+n)\bigr)=n .
\]
Consequently $\tau\bigl(\Gamma_{\mathcal{G}_p}(\mathcal{I})\bigr)\leq n$, and by
$\eqref{inequalities}$,
\[
  \sum_{i\in\mathcal{I}}\hat{\lambda}_i\leq n .
\]
Finally, if $(k+1)\in\mathcal{I}$, then by Part \eqref{thmr2} of Theorem \ref{thmR1} the $n$
pairwise disjoint recovery sets $\{j,j+n\}$, $j\in[n]$, give
$\hat{\lambda}_{k+1}=n$, so the bound is achievable.
\end{proof}

\begin{example}\label{exR2}
    Consider $\mathcal{C}$ as given in Example \ref{exR1}. 
For every $I\subseteq[k=3],$ we have $\nu(\Gamma_\mathcal{G}(I))=\tau_f(\Gamma_\mathcal{G}(I))=3.$ Consequently, the SRR of $\mathcal{C}$ is \[\Lambda(\mathcal G)=\{(\lambda_1,\lambda_2,\lambda_3)\in\mathbb{R}^3_{\geq 0}\mid \lambda_1+\lambda_2+\lambda_3\leq 3\}.\]
Now consider \(\mathcal{C}_p\). The maximal achievable service rate for $x_{4}$ is $\hat{\lambda}_4^*=6.$ For all $I\subseteq[3]$, we have $\nu(\Gamma_{\mathcal{G}_p}(I))=\tau_f(\Gamma_{\mathcal{G}_p}(I))=6.$ Consequently, $\hat{\lambda}_i^*=6$, and $\sum\limits_{i\in I}\hat{\lambda}_i\leq 6.$ Moreover, in this example, for any $I\subseteq[4]$, we have $\sum\limits_{i\in I}\hat{\lambda}_i\leq 6.$ Therefore,  \[
\Lambda(\mathcal G_p)=\Bigl\{(\hat\lambda_1,\dots,\hat\lambda_4)\in\mathbb R^4_{\ge0}\ \mid\
\hat\lambda_1+\hat\lambda_2+\hat\lambda_3+\hat\lambda_4\le 6\Bigr\}.
\]
\end{example}
 Similar results can be inductively proven for the iterated Plotkin-type codes. 
    
\begin{theorem}\label{Cpmvc}
     Let $\mathcal{C}$ be an $[n,k,d]$ binary linear code with a generator matrix $\mathcal{G}$, and $\mathcal{C}_{p^m}$ be the iterated Plotkin-type code defined as in \eqref{Cp^m}. Let $I\subseteq[k]$, and $\hat{\boldsymbol{\lambda}}\in\Lambda(\mathcal{G}_{p^m})$. Then \[
\sum_{i\in I}\hat{\lambda}_i\leq \tau_f\bigl(\Gamma_{\mathcal{G}_{p^m}}(I)\bigr)
\leq
2^{m}\tau_f\bigl(\Gamma_{\mathcal{G}}(I)\bigr).
\]
\end{theorem}
\begin{theorem}\label{CpmI}
     Let $\mathcal{C}$ be an $[n,k,d]$ binary linear code with a generator matrix $\mathcal{G}$, and let $\mathcal{C}_{p^m}$ be the iterated Plotkin-type code defined as in \eqref{Cp^m}. Let $\mathcal{I}\subseteq[k+1]$, and let $\hat{\boldsymbol{\lambda}}\in\Lambda(\mathcal{G}_{p^m})$. Then
 $\sum_{i\in \mathcal{I}}\hat{\lambda}_i\leq  n2^{m-1}$, and this bound is tight.
\end{theorem}

\subsection{Application to first-order binary Reed-Muller codes}
    The first-order binary Reed-Muller codes $R(1,m)$  with parameters $[2^m,m+1,2^{m-1}]$  have the following inductive definition with generator matrix $\mathcal{G}_m$ \cite{Ling}:\[R(1,m+1)=\{(\mathbf{c},\mathbf{c}) : \mathbf{c}\in R(1,m)\}\cup\{(\mathbf{c},\mathbf{c}+1) : \mathbf{c}\in R(1,m)\},\]
 which is exactly the Plotkin-type construction of Section \ref{plotkinc}. Consider an information vector $(x_1,x_2)$ and a generator matrix for $\mathcal{R}(1,1)$ 
\[
  \mathcal{G}_1=\begin{pmatrix}1&1\\0&1\end{pmatrix}.
\]
 The minimal recovery sets and the achievable service rates of the two information symbols are as follows:\[\mathcal{R}_1(\mathcal{G}_1)=\{(1)\},\mathcal{R}_2(\mathcal{G}_1)=\{(1,2)\},\text{ and } \lambda_1\leq1,\lambda_
2\leq1.\] 
Writing
$\mathcal{G}_m$ for the generator matrix obtained from $\mathcal{G}_1$ by $m-1$ successive applications of this construction, we have
$$\mathcal{G}_{m}=\begin{pmatrix}
    \mathcal{G}_{m-1}\ &\mathcal{G}_{m-1}\\\mathbf{0}_{2^{m-1}}&\mathbf{1}_{2^{m-1}}
\end{pmatrix}$$ and the results of Sections \ref{sec3} and \ref{sec4} apply at every stage. In the notation of
those sections, the base code at the $m$-th stage is $R(1,m-1)$, of length
$n_{m-1}=2^{m-1}$ and dimension $m$, and the newly adjoined information symbol in $R(1,m)$
is $x_{m+1}$.

Then, we have $R(1,2)$ with the generator matrix $\mathcal{G}_2$, information vector $(x_1,x_2,x_3)$, and the associated recovery sets as $\mathcal{R}_1=\{(1),(2,3,4)\},\mathcal{R}_2=\{(1,2),(3,4)\},\mathcal{R}_{3}=\{(1,3),(2,4)\}$.
Moreover, for each nonempty subset $I\subseteq[3],$ it follows that $\sum_{i\in I}\lambda_i\leq \nu(\Gamma_{\mathcal{G}_2}(I))=\tau_f(\Gamma_{\mathcal{G}_2}(I))=2$,  i.e.,\[\Lambda(\mathcal{G}_2)=\{(\lambda_1,\lambda_2,\lambda_3)\in\mathbb{R}^{3}_{\ge0}: \sum_{i=1}^3\lambda_i\leq2\}.\]

   \subsubsection*{The symbols $x_i$, $i\ge 2$, for $m\ge 3$}
   
   Fix $m\geq 3$. The symbol $x_{m+1}$ is the one adjoined at the last stage, so by
Part \eqref{thmr2} of Theorem \ref{thmR1} it possesses $n_{m-1}=2^{m-1}$ pairwise
disjoint recovery sets of cardinality $2$, and by \eqref{new},
\[
  \lambda_{m+1}^{*}=\nu\bigl(\Gamma^{m+1}_{\mathcal{G}_m}\bigr)=
  \tau\bigl(\Gamma^{m+1}_{\mathcal{G}_m}\bigr)=2^{m-1}.
\]
For $2\le i\le m$ the symbol $x_i$ was adjoined at the $i$-th stage, where it
likewise acquired $2^{i-2}$ pairwise disjoint recovery sets of cardinality $2$;
applying Corollary \ref{colR11}  once at each of the subsequent $m-i+1$ stages, we have that $\mathcal{R}_i(\mathcal{G}_m)$ contains $2^{m-1}$
pairwise disjoint recovery sets and $\nu(\Gamma^{i}_{\mathcal{G}_m})\ge 2^{m-1}$.
In the converse direction, Theorem \ref{Cpmvc} applied along the chain from
$\mathcal{G}_2$ gives
$\tau_f(\Gamma^{i}_{\mathcal{G}_m})\le 2^{m-2}\tau_f(\Gamma^{i}_{\mathcal{G}_2})
=2^{m-2}\cdot 2=2^{m-1}$. By \eqref{inequalities},
\[
  \lambda_i^{*}=2^{m-1}\qquad\text{for all } i\in[2,m+1].
\]

\subsubsection*{The symbol $x_1$}
The information symbol \(x_1\) is the only symbol carried over from \(\mathcal{G}_1\) and the only one with a systematic column in \(\mathcal{G}_m\). In \(\mathcal{G}_2\), it has a unique recovery set of cardinality $3$. By Theorem \ref{thmI1}, this recovery set propagates through the subsequent stages, so that after \(m-2\) further stages, there are 
\[
  N_m=\frac{2^{2m-1}-3\cdot 2^{m-1}+1}{3}=\frac{(2^{m}-1)(2^{m-1}-1)}{3}
\]
recovery sets of cardinality $3$ for $x_1$ in $\mathcal{R}_1(\mathcal{G}_m)$, none of which contains the systematic server, and each remaining server lies in exactly $2^{m-1}-1$ of them. Assigning the rate $1/(2^{m-1}-1)$ to each of them
loads every non-systematic server to exactly $1$, and together with unit rate
from the systematic server,
\[
  \lambda_1^{*}\;\ge\;1+\frac{N_m}{2^{m-1}-1}\;=\;1+\frac{2^{m}-1}{3}
  \;=\;\frac{2^{m}+2}{3}.
\]

 The upper bound given in \cite{Ly} for the maximal achievable service rate gives $$\lambda_1^*\leq1+\frac{2^m-1}{3}=\frac{2^{m}+2}{3}.$$ 
\subsubsection*{Cumulative service rate}
By Theorem \ref{CpmI},  the cumulative service rate of data symbols indexed by a subset $I\subseteq[m+1]$ satisfies
\begin{equation}\label{RMsum}  \sum_{i\in I}\lambda_i\;\le\;2^{m-1}.
\end{equation}

To summarize, the service rate region of first-order binary Reed-Muller codes R($1,m$), for $m\geq3$, is contained in the set of all nonnegative $\boldsymbol{\lambda}=(\lambda_1, \lambda_2,\ldots, \lambda_{m+1})$ satisfying-

\[\begin{cases}\lambda_i\geq 0 ~\forall~i\in[m+1],\\
    \lambda_1\leq\frac{2^{m}+2}{3},\\ 
    \lambda_i\leq2^{m-1} ~\forall~i\in[2,m+1],\\
    \sum\limits_{i\in I}\lambda_i\leq 2^{m-1} ~\text{ for }I\subseteq[m+1].
\end{cases}\]

\begin{lemma}\label{lemRMcoupling}
For every $m\ge 2$ and every
$\boldsymbol\lambda\in\Lambda(\mathcal{G}_m)$,
\[
  3\lambda_1+2\sum_{i=2}^{m+1}\lambda_i\;\le\;2^{m}+2 .
\]
\end{lemma}

\begin{proof}
By Theorem \ref{thmR1}, every $R_p\in\mathcal{R}_i(\mathcal{G}_m)$ has
$|R_p\cap[2^{m-1}+1,2^{m}]|$ even for $i\le m$ and odd for $i=m+1$; in
particular no recovery set for $x_i$ with $i\ge 2$ is a singleton, since
$\mathcal{G}_m$ has a systematic column only for $x_1$. Hence every minimal
recovery set for $x_i$, $i\ge 2$, has cardinality at least $2$. For $x_1$, the
unique recovery set of cardinality $1$ is the systematic one, and by Theorems
\ref{thmR2} and \ref{thmR8} all others have cardinality at least $3$, since
$d(\mathcal{C}_p^{\perp})=\min\{d(\mathcal{C}^{\perp}),4\}=4$ for every stage
from the second onward.

For a feasible allocation, let $\lambda_{1,s}\le 1$ denote the rate assigned to the
systematic node for $x_1$. Summing the loads of
the $2^{m}$ servers gives
\[
  \lambda_{1,s}+3(\lambda_1-\lambda_{1,s})+2\sum_{i=2}^{m+1}\lambda_i\;\le\;2^{m},
\]
and the claim follows.
\end{proof}

\begin{theorem}\label{thmRM}
For every $m\ge 2$,
\[
  \Lambda(\mathcal{G}_m)=\Bigl\{\boldsymbol\lambda\in\mathbb{R}^{m+1}_{\ge 0}:
  \sum_{i=1}^{m+1}\lambda_i\le 2^{m-1},\;\;
  3\lambda_1+2\sum_{i=2}^{m+1}\lambda_i\le 2^{m}+2\Bigr\}.
\]
\end{theorem}

\begin{proof}
Write $\mathcal{P}_m$ for the set on the right. The inclusion
$\Lambda(\mathcal{G}_m)\subseteq\mathcal{P}_m$ is \eqref{RMsum} together with
Lemma \ref{lemRMcoupling}.

For the reverse inclusion, observe that $\mathcal{P}_m$ is defined by two non-trivial inequalities in $\mathbb{R}^{m+1}_{\ge0}$, so each of its
vertices has at most two nonzero coordinates. Additionally, these coordinates can not both have indices at least $2$, as equality in both non-trivial inequalities would give
$2\cdot 2^{m-1}=2^{m}+2$. The vertices are therefore
\[
  \mathbf{0}_{m+1},\qquad \tfrac{2^{m}+2}{3}\mathbf{e}_1,\qquad
  2^{m-1}\mathbf{e}_i,\qquad 2\mathbf{e}_1+(2^{m-1}-2)\mathbf{e}_i
  \qquad (\text{ where }i\in[2,m+1]).
\]
Convexity of \(\Lambda(\mathcal{G}_m)\) reduces the proof to realizing these vertices. The first three have already been obtained. 

For the fourth case, fix $i\ge 2$. By Part \eqref{thmr2} of Theorem \ref{thmR2} and Corollary \ref{colR11}, there are $2^{m-1}$ pairwise disjoint recovery sets of cardinality $2$ for $x_i$, say $S_1, S_2,\ldots, S_{2^{m-1}}$. Let $S_1=\{1,\varepsilon\}$ be the recovery set for $x_i$ that contains the systematic server for $x_1$. Each of the other $S_t$ ($t\neq1$) together with $\varepsilon$, recovers $x_1$. Moreover, $(S_a\cup\{\varepsilon\})\cap (S_b\cup\{\varepsilon\})=\{\varepsilon\}$, where $2\leq a\neq b\leq 2^{m-1}$.


Assign load $1$ to set $\{1\}$ and load $\frac{1}{2^{m-1}-1}$ to the sets $S_t\cup\{\varepsilon\}$ for all $2\leq t\leq 2^{m-1}$. Then $\lambda_1=2$,
the servers first and $\varepsilon$ carry load exactly $1$. Assigning the load $\frac{2^{m-1}-2}{2^{m-1}-1} $ to the pair $S_t$ for
$t\in[2,2^{m-1}]$ and zero to $S_1$, exhausts every server's capacity and gives
\[
  \lambda_i=\frac{2^{m-1}-2}{2^{m-1}-1}(2^{m-1}-1)=2^{m-1}-2 .
\]
\end{proof}

 \section{Conclusion}\label{secC}
 This work analyzes how standard code construction operators act on the SRR, given that the SRR for the parent code is known. We establish a correspondence between the recovery set structure of a Plotkin-type code $\mathcal {C}_{p} $ and the underlying parent code $\mathcal{C}$, and consequently a correspondence between their achievable service rates. As a consequence of the obtained results, we derive a bound for service rates of iterated Plotkin constructions and subsequently obtain the service rates for the binary first-order Reed--Muller code $R(1,m)$, for any $m$, using a generator matrix for $R(1,2)$. As future work, it would be interesting to examine the relationship of the general 
$(\mathbf u\mid\mathbf u+\mathbf v)$ construction with an arbitrary second code, and to extend the analysis to \(q\)-ary codes and other code constructions.
 \section*{Acknowledgments}
 The first author is grateful to the University Grants Commission (UGC), Government of India, for the financial support. 

{}

\end{document}